\documentclass[11pt]{article}

\usepackage[numbers]{natbib}
\usepackage{amsmath}
\usepackage{amssymb}
\usepackage{amsthm}
\usepackage{xcolor}
\usepackage{hyperref}
\usepackage{algorithm}
\usepackage{algpseudocode}
\usepackage[margin=1in]{geometry}

\hypersetup{colorlinks=true,citecolor=blue,linkcolor=blue}
\allowdisplaybreaks

\newcommand{\paperTitle}{Hadamard Flattening and Gaussian Pooling Sketch for Least Squares with Coordinate-wise Guarantee}
\newcommand{\paperAuthor}{Zhao Song\thanks{Independent Researcher. Email: \texttt{magic.linuxkde@gmail.com}.} \and Lichen Zhang\thanks{Massachusetts Institute of Technology. Email: \texttt{lichenz@mit.edu}.}}

\usepackage{graphicx}
\usepackage{caption}
\newcommand{\ind}{\mathrel{\perp\!\!\!\perp}}
\newcommand{\Tmat}{{\cal T}_{\mathrm{mat}}}
\newcommand{\R}{\mathbb{R}}
\newcommand{\E}{\mathbb{E}}
\newcommand{\op}{{\rm op}}

\theoremstyle{plain}
\newtheorem{theorem}{Theorem}[section]
\theoremstyle{definition}
\newtheorem{lemma}[theorem]{Lemma}
\theoremstyle{plain}
\newtheorem{claim}[theorem]{Claim}
\newtheorem{definition}[theorem]{Definition}
\newtheorem{proposition}[theorem]{Proposition}

\theoremstyle{definition}

\begin{document}

\date{}
\title{\paperTitle}
\author{\paperAuthor}
\maketitle

\begin{abstract}
Randomized sketch-and-solve algorithms accelerate overconstrained $\ell_2$
regression by replacing the input with a much smaller, randomly projected
problem. While standard subspace embeddings guarantee that the cost of the
regression is nearly preserved, coordinate-wise accuracy of the solution is
more delicate. Instead of preserving the objective value, we want the solution
vector itself to be close to the optimal solution in $\ell_\infty$ norm.
In particular, we want to find a vector $x'\in \mathbb{R}^d$ such that
$\|x'-x^\star\|_\infty\leq \frac{\epsilon}{\sqrt d}\cdot \|Ax^\star-b\|_2\cdot
\|A^\dagger\|_{\op}$, where $A\in \mathbb{R}^{n\times d}$ is the design
matrix, $b\in \mathbb{R}^n$ is the label vector, $x^\star\in \mathbb{R}^d$ is the
optimal solution and $A^\dagger$ is the pseudo-inverse of $A$. Price, Song and
Woodruff initiated the study of this problem and claimed that the subsampled
randomized Hadamard transform (SRHT) with
$O(\epsilon^{-2} d^{1+\Theta(\sqrt{\log\log n/\log d})})$ rows achieves this
guarantee. We point out a gap in the proof of this claim. A subsequent work of
Song, Ye, Yin and Zhang claimed a smaller row count of
$O(\epsilon^{-2}d\log^3 n)$. However, their proof relies on an independence
assumption that does not hold in general, and we exhibit an explicit instance
on which it fails.

To achieve a truly nearly-linear-in-$d$ row count, we introduce a new fast,
dense randomized transform, which combines a randomized Hadamard flattening, a
random permutation, and balanced, disjoint Gaussian pooling. Conditioned on
the Hadamard-and-permutation stage, the sketched problem becomes an exact
Gaussian regression in which the noise is independent of the entire sketched
design. This conditional independence is exactly what the argument of Song,
Ye, Yin and Zhang was missing. Our sketch yields the $\ell_\infty$ guarantee
with $m=O(\epsilon^{-2}d\log d)$ rows, using one Hadamard pass with a padded
internal dimension $N=\widetilde{O}(n+\epsilon^{-2}d^3)$. The sketch is also
efficient to apply, since the sketched pair $(SA, Sb)$ can be computed in
$O(Nd\log N)=\widetilde{O}(nd+\epsilon^{-2}d^4)$ time.

\end{abstract}

\section{Introduction}

Least-squares regression is a standard tool in scientific computing, numerical
linear algebra, and modern machine learning pipelines. In the overconstrained
setting, given a data matrix $A\in\R^{n\times d}$ with $n\gg d$ and a label
vector $b\in\R^n$, we wish to (approximately) minimize $\|Ax-b\|_2$ over
$x\in\R^d$. A natural way to speed up this minimization is to compress the
problem before solving it. In the sketch-and-solve paradigm, we draw an
oblivious random matrix $S\in\R^{m\times n}$ with $m\ll n$. We then form $SA$
and $Sb$ and solve the $m$-row problem $\min_{x\in\R^d}\|SAx-Sb\|_2$ instead.
Randomized compression of this kind traces back to the Johnson--Lindenstrauss
lemma~\cite{jl84}: any $N$ points in Euclidean space can be mapped into
$O(\epsilon^{-2}\log N)$ dimensions so that every pairwise distance is
preserved up to a $1\pm\epsilon$ factor. Larsen and Nelson~\cite{ln17} proved
that no embedding, linear or not, can do better in the worst case. For
regression, which asks for more than pairwise distances, the right tool is an
oblivious subspace embedding~\cite{sar06,dmms11}. Such an embedding is a
random matrix that simultaneously preserves the norm of every vector in a
fixed low-dimensional subspace, which suffices for relative-error guarantees
on the objective. Oblivious subspace embeddings admit fast implementations of
two kinds, dense ones built from randomized Hadamard transforms and sparse
ones built from hashing. The sparse ones include CountSketch, which originated
in streaming frequency estimation~\cite{ccfc02} and later enabled regression
in input-sparsity time~\cite{cw17}. Another sparse embedding,
OSNAP~\cite{nn13}, refined the tradeoff between sparsity and embedding
dimension. Recently, OSNAP has been proven to be optimal. Heidary~\cite{h26}
and, independently, Mai and Rao~\cite{mr26} showed that the variant of OSNAP
that stacks independent copies of CountSketch attains the embedding dimension
and sparsity conjectured by Nelson and Nguyen.

However, guarantees of this type primarily control the residual or a global
norm of the solution error. When $S$ is an oblivious subspace embedding with
distortion $\epsilon$, the minimizer $x'$ of the sketched problem satisfies
$\|x'-x^\star\|_2\leq O(\epsilon)\cdot\|Ax^\star-b\|_2\cdot\|A^\dagger\|_{\op}$,
where $x^\star$ denotes the optimal solution. For a direction $a\in\R^d$ fixed
independently of the sketch, we then obtain by Cauchy--Schwarz
\begin{align*}
|\langle a,x'-x^\star\rangle|\leq O(\epsilon)\cdot\|a\|_2\cdot\|Ax^\star-b\|_2\cdot\|A^\dagger\|_{\op}.
\end{align*}
This estimate loses a factor of $\sqrt{d}$: for a generic direction
$a$, the deviation should scale like a $1/\sqrt{d}$ fraction of the $\ell_2$
error. A bound of this refined form, applied simultaneously to
$a=e_1,\ldots,e_d$, yields an $\ell_\infty$ guarantee on $x'-x^\star$.
Price, Song and Woodruff~\cite{psw17} claimed a bound of this form for the
subsampled randomized Hadamard transform (SRHT), albeit with a row count that
is superlinear in $d$. In Section~\ref{sec:psw17_gap}, we point out a gap in
the proof of their bound. Later, Song, Ye, Yin and Zhang~\cite{syyz23} claimed a nearly-linear-in-$d$
row count by combining a constant-distortion oblivious subspace embedding with
an oblivious coordinate-wise embedding for a fixed pair of vectors.
Unfortunately, their argument applies the fixed-vector guarantee to a
direction that contains the inverse of the random sketched Gram matrix and
therefore depends on the same sketch. The stated guarantee does not cover this
adaptive choice, and as we show in Section~\ref{sec:original_dependency_bug},
the gap is genuine rather than merely formal.

In this work, we introduce a new sketch, which we call the \emph{balanced
Gaussian-pooled transform}. We prove that it achieves the $\ell_\infty$
guarantee with a truly nearly-linear-in-$d$ row count. Since this sketch is
the central object of the paper, we first state its ingredients explicitly.
The transform factors as
\begin{align*}
S:=BD_\gamma \Pi F_N D_\sigma J,
\end{align*}
applied from right to left. Here $J$ pads the input from dimension $n$ to a
larger internal dimension $N$ by appending zeros, $D_\sigma$ is a diagonal
matrix of random signs, $F_N$ is one normalized Hadamard transform, and $\Pi$
is a uniformly random permutation. The matrix $D_\gamma$ is a diagonal matrix
of independent standard Gaussians, and $B$ partitions the $N$
coordinates into $m$ consecutive blocks of equal size and sums each block. We
call the block-summing step \emph{Gaussian pooling}, since each row of the
sketch pools an entire block of Gaussian-weighted coordinates. The blocks are
\emph{balanced} in the sense that, with high probability, the Hadamard
transform and the random permutation spread the subspace relevant to the
regression so evenly that every block carries a nearly equal share of its
energy.

The intuition behind the construction is to expose the randomness in two
stages. The first stage consists of the randomized Hadamard transform followed
by the random permutation. This stage flattens and shuffles the (padded)
column space of $A$ together with the residual direction, so that every block
becomes a nearly isotropic copy of this fixed $(d+1)$-dimensional subspace.
Conditional on the first stage, the second stage takes over: because the
pooling blocks are disjoint, the rows of the sketched problem are independent
Gaussian vectors. Under the same conditioning, the sketched least squares
becomes an exact Gaussian regression whose noise is independent of the entire
design. The inverse sketched Gram matrix may then depend on the design in an
arbitrary fashion, with no fixed-vector guarantee ever applied to a
sketch-dependent direction. We prove that, with probability at least
$1-\delta$, the resulting estimator satisfies the fixed-direction guarantee
with $m=O(\epsilon^{-2}d\log(1/\delta))$ rows and the simultaneous guarantee
over all $d$ coordinates with $m=O(\epsilon^{-2}d\log(d/\delta))$ rows.
Although the transform uses a single Hadamard pass, it is not the canonical
SRHT: the analysis works with a padded internal dimension
$N=\widetilde{O}(n+\epsilon^{-2}d^3)$.

\subsection{Main Result}

Given $A\in\R^{n\times d}$ of full column rank and $b\in\R^n$, define the
exact and the sketched least-squares solutions by
$
  x^\star:= {\arg\min}_x \|Ax-b\|_2,
 $
  $\widehat{x}:= {\arg\min}_x \|SAx-Sb\|_2
$.
For a fixed direction $a\in\R^d$, the target is
\begin{equation}
  |a^\top(\widehat{x}-x^\star)|
  \leq
  \frac{\epsilon}{\sqrt{d}}\,
  \|a\|_2\,\|Ax^\star-b\|_2\,\|A^\dagger\|_{\op},
  \label{eq:target}
\end{equation}
which gives the corresponding $\ell_\infty$ estimate after taking $a=e_j$ and
union bounding over $j\in[d]$.

Our main result is the following theorem.

\begin{theorem}[Main result]\label{thm:main}
Fix $0<\epsilon\leq 1$ and $0<\delta<1/2$. There is a distribution over
oblivious linear maps $S\in\R^{m\times n}$, each applicable with a single
randomized Hadamard transform, such that for every fixed $A,b,a$, a single
draw of $S$ satisfies Eq.~\eqref{eq:target} with probability at least
$1-\delta$ when
$
  m=O(\epsilon^{-2}d\log(1/\delta))
$.
The guarantee holds for all $d$ coordinates simultaneously using
$
  m=O(\epsilon^{-2}d\log(d/\delta))
$
rows, with the failure parameter $\delta/d$ used throughout the construction. The map is
the balanced Gaussian-pooled transform
(Definition~\ref{def:balanced_gaussian_pooled_sketch}), whose padded
internal dimension $N$ and total classical runtime $T$ are
\[
  N=\widetilde{O}(n+\epsilon^{-2}d^3),
  \qquad
  T=O(Nd\log N+\Tmat(d,m,d)+d^\omega)
   =\widetilde{O}(nd+\epsilon^{-2}d^4).
\]
Here $\Tmat(a,b,c)$ denotes the arithmetic time required to multiply an
$a\times b$ matrix by a $b\times c$ matrix, and $\widetilde{O}$
suppresses logarithms in $n,d,1/\epsilon,1/\delta$.
\end{theorem}

Price, Song and Woodruff~\cite{psw17} claim in their Theorem~10 that, under
their assumptions, the canonical SRHT satisfies Eq.~\eqref{eq:target} with
\[
  m=
  O(
    \epsilon^{-2}d^{1+\Theta(\sqrt{\log\log n/\log d})}
  )
\]
rows and polynomially small failure probability. The proof of that theorem has
a gap in the second-moment computation of their Lemma~9, which keeps only the
pairs of identical index tuples, while the pairs it drops are not small. The
parts of their proof that do not use Lemma~9 already give
Eq.~\eqref{eq:target} with high probability for
$m=O(\epsilon^{-2}d\operatorname{polylog}(n))$ rows when
$\epsilon\leq d^{-1/2}$. For $\epsilon>d^{-1/2}$, a row count of
$\epsilon^{-2}d\operatorname{polylog}(n)$ for the canonical SRHT would need an
argument beyond this paper. We locate the gap and give exact instances in
Section~\ref{sec:psw17_gap}. The statement of their Theorem~10 may still be
true.

Song, Ye, Yin and Zhang~\cite{syyz23} claim in their Theorem~1.1 that the
canonical SRHT satisfies Eq.~\eqref{eq:target} with
$m=O(\epsilon^{-2}d\log^3(n/\delta))$ rows. As we explain in
Section~\ref{sec:original_dependency_bug}, the proof of that claim is
incomplete, and the issue cannot be repaired within the same argument.

\paragraph{Roadmap.}
In Section~\ref{sec:prior_gaps}, we locate a gap in the proof of the SRHT bound
of~\cite{psw17} (Section~\ref{sec:psw17_gap}), describe the dependency issue in
the argument of~\cite{syyz23} (Section~\ref{sec:original_dependency_bug}), construct an explicit instance
showing that the gap is genuine (Section~\ref{sec:gap_genuine}), and show that the natural fix of
replacing the Rademacher diagonal in the SRHT with a Gaussian diagonal fails as
well (Section~\ref{sec:gaussian_diagonal_fails}). In
Section~\ref{sec:balanced_gaussian_pooled_transform}, we define the balanced
Gaussian-pooled transform and explain the conditional-regression viewpoint that
repairs the argument, and we show that, conditional on the Hadamard-and-permutation
stage, the pooled rows are independent Gaussian vectors with explicit covariances
(Section~\ref{sec:conditional_pooling}). In Section~\ref{sec:core_lemma_proof},
we develop the block geometry supplied by the randomized Hadamard transform and
the exact conditional Gaussian regression representation, prove concentration
for the conditional design and bias, record that our transform is an oblivious
subspace embedding, and prove the corrected core lemma. In
Section~\ref{sec:consequences}, we transfer the corrected core lemma to least
squares and bound the running time of our sketch.

\section*{Acknowledgement \& AI Disclosure}

The sketch construction is developed by GPT Codex 5.6 Sol and Claude Code Fable 5. Proofs are generated part by GPT Codex 5.6 Sol and authors, and are verified independently by Claude Code Fable 5 and authors. The authors take full responsibility for verifying all claims and for the paper's content.

\section{Gaps in Prior Work}
\label{sec:prior_gaps}

\subsection{The SRHT Bound of Price, Song and Woodruff}
\label{sec:psw17_gap}

Price, Song and Woodruff~\cite{psw17} studied the coordinate-wise guarantee for
the SRHT, which we denote by $S=m^{-1/2}PHD$. Here $H$ is the unnormalized
Sylvester Hadamard matrix, $H_{ij}=(-1)^{\langle i-1,j-1\rangle}$ with indices
in binary, $P$ samples $m$ rows and $D$ is a Rademacher diagonal. This
normalization gives $\langle S_i,S_i\rangle=1$, as their proof of Lemma~9
uses~\cite[p.~9]{psw17}.\footnote{Page, equation and appendix numbers of~\cite{psw17}
refer to its full version, arXiv:1705.10723v1.} Their Theorem~10 states that
Eq.~\eqref{eq:target} holds with $m=O(\epsilon^{-2}d^{1+\alpha})$ rows for
$\alpha=\Theta(\sqrt{\log\log n/\log d})$ when
$\log n=d^{o(1)}$~\cite[p.~8]{psw17}. Unlike the argument
of~\cite{syyz23}, its top-level argument applies approximate matrix product and
subspace embeddings only to objects fixed before the sketch is drawn. We find that
the proof of Theorem~10 has a gap, located in the proof of their Lemma~9, while
the statement of Theorem~10 may still be true.

\paragraph{The gap.}
To compute a second moment, the proof of Lemma~9 keeps only the pairs of
identical index tuples, while the pairs it drops are not small. With
$A=U\Sigma V^\top$ and $T:=I-U^\top S^\top SU$, Claim~8 of~\cite[p.~7]{psw17}
expands $(SA)^\dagger S$ as their series~(7),
$V\Sigma^{-1}(\sum_{j\geq 0}T^j)U^\top S^\top S$, whenever $\|T\|\leq 1/2$.
Let $\beta:=\Sigma^{-1}V^\top a$ for the direction $a$ of
Eq.~\eqref{eq:target} and take $b$ with $U^\top b=0$ as their proof does. We
define $g_j:=(-1)^j\beta^\top T^jU^\top S^\top Sb$ so that $(-1)^jg_j$ is the
$T^j$ term of this series applied to $b$ and tested against $a$. For a fixed
$P$, each $g_j$ is a polynomial of degree at most $2j+2$ in the signs of $D$.
The proof of Theorem~10 needs a bound on $|g_j|$ for every $1\leq j<k$, where
$k$ is the cutoff of that proof, holding with high probability over $P$ and the
signs. By hypercontractivity, it suffices to bound $\E_\sigma[g_j^2]$ with
high probability over $P$, up to a factor $(C\log n)^{O(j)}$. The proof takes
these bounds from Lemma~9, after writing the $T^j$ term as a sum of $2^j$
terms of the form in that lemma~\cite[p.~9]{psw17}. Since the proof of Lemma~9
expands each such term over tuples of coordinates, every pair of tuples appears
when we square such a term. A pair adds to the second moment whenever each sign
occurs in it an even number of times. The proof keeps only the pairs of
identical tuples, in Eq.~(10) on p.~10 and Eq.~(17) on p.~24
of~\cite{psw17}. The Khintchine steps of pp.~10--11 then turn the result into a
tail bound. We show below that the exact identity
$S^\top S=I+DED$, where $E$ has zero diagonal, removes every tuple in which the
two indices of a factor $\langle S_{i_c},S_{j_c}\rangle$ coincide, in which
case the factor is~$1$. The omitted pairs that remain are pairs of two
different tuples
in which each sign occurs an even number of times. These pairs are not small,
since on the $n=8$ instance below they raise $\E_\sigma[g_j^2]$ above the
identical-tuple sum by the factors $6/5$, $4$ and $328/15$ for $j=1,2,3$. A
correct bound on these remaining pairs would close the gap in the proof. The
identical-tuple sum is at most, term by term, the sum that Appendix~E
of~\cite{psw17} bounds on pp.~24--27. The exact identity removes the tuples
with coinciding indices, and hypercontractivity replaces the Khintchine steps.

The proof of Theorem~10 bounds the term $j=0$ by approximate matrix product
(their bound~(8)) and the terms with $j\geq k$ by a geometric tail bound. The
remaining terms are bounded with
Lemma~9, which controls
$f:=v^\top S^\top S(UU^\top S^\top S)^{\ell}b$ for a unit vector $v$ and an
integer $\ell\geq0$ (their $k$)~\cite[pp.~7--9]{psw17}. Since Claim~8, the
bound~(8) and the tail bound in the proof of Theorem~10 do not use Lemma~9, the
gap affects only the terms bounded with Lemma~9.

Using only Claim~8 and these two bounds, with the cutoff $k=1$ and summing the
display of p.~8 directly, we already obtain Eq.~\eqref{eq:target} in part of
the range of Theorem~10. For every $j\geq1$, their display on p.~8 bounds the $T^j$ term by
$\gamma^j\sqrt{d/m}\,\|a\|_2\|\Sigma^{-1}\|\|b\|_2$ up to polylogarithmic
factors, where $\gamma:=\sqrt{d\log^c n/m}$ is the distortion of the subspace
embedding. Summing over $j\geq1$, we obtain $O(d/m)$ up to the same factors,
which is at most $\epsilon/\sqrt d$ once $m\geq
C\epsilon^{-1}d^{3/2}\operatorname{polylog}(n)$. The bound~(8) needs $m\geq
C\epsilon^{-2}d\operatorname{polylog}(n)$. Hence, for
$\epsilon\leq d^{-1/2}$, the SRHT with $m=O(\epsilon^{-2}d\operatorname{polylog}(n))$
rows satisfies Eq.~\eqref{eq:target} with high probability without Lemma~9, and
the gap matters only for larger $\epsilon$.

The proof of Lemma~9 expands $f$ over index tuples as
$f=\sum_{i_0,j_\ell}\sigma_{i_0}\sigma_{j_\ell}Z_{i_0,j_\ell}$, where $\sigma$ is
the diagonal of $D$ and each $Z_{i_0,j_\ell}$ is a polynomial of degree $2\ell$
in the same signs~\cite[p.~9]{psw17}. The identical-tuple sum of Eq.~(10)
and Eq.~(17) is neither an upper nor a lower bound on the second moment of
$Z_{i_0,j_\ell}$, as we show on the $n=4$ instance below. The passage from $f$
to $\|Z\|_F$ is not written out, and the Khintchine steps around it treat the
entries of $Z$ as fixed coefficients~\cite[pp.~10--11]{psw17}. These steps
need
each column of $Z$ to be independent of the whole sign vector that it is summed
against, which fails.

We close part of the gap by a standard argument. Let $\widetilde
S:=m^{-1/2}PH$, so that $S=\widetilde SD$. With $\widetilde S^\top\widetilde
S=I+E$, where $E$ has zero diagonal, we have $S^\top S=I+DED$ exactly.
For $U^\top b=0$ and for $v=U\beta$ up to normalization, this identity gives
\[
  f=\sum_{j=0}^{\ell}\binom{\ell}{j}g_j,
  \qquad
  g_j=\beta^\top K^jU^\top DEDb,
  \qquad
  K:=U^\top DEDU,
\]
which agrees with the definition of $g_j$ above because $T=-K$. The expansion
thus reverses the step on p.~9 of~\cite{psw17} that writes the $T^j$
term through Lemma~9. For
each fixed $P$, the hypercontractive inequality for polynomials in
Rademacher signs gives $|g_j|\leq(C\log n)^{j+1}(\E_\sigma[g_j^2])^{1/2}$ with
probability $1-n^{-c}$ over the signs, which replaces the Khintchine steps.
Bounding $\E_\sigma[g_j^2]$ is the quantitative core of Lemma~9 and not a small
residual, since it amounts to bounding the $T^j$ terms of their series~(7)
directly. For the expansion we need $v$ in the column span of $U$, which holds
where Theorem~10 uses Lemma~9, while Lemma~9 is stated for every unit vector
$v$.

We now give exact values on small instances. For $n=4$, let $P$ sample rows
$1,2,3$ of the $4\times4$ Sylvester Hadamard matrix once each and take $d=1$,
$U=v=\frac12(1,1,1,1)^\top$, $b=\frac12(1,-1,1,-1)^\top$ and
$(i_0,j_\ell)=(1,2)$. Then $\E_\sigma[Z_{1,2}^2]$ equals $17/2592$ for $\ell=1$
and $79/7776$ for $\ell=2$, while Eq.~(17) gives $1/144$ and $1/432$. In this
instance we have $g_1=g_2=0$ for every sign vector, so $f=g_0$ exactly. This
discrepancy therefore bears only on the term $j=0$ that the bound~(8)
handles, while for the terms $g_j$ with $j\geq1$ we need a larger instance.
For $n=8$, let $P$ sample rows $1,2,3,5$ of the $8\times8$ Sylvester Hadamard
matrix once each, so that $m=4$ and every entry of $E$ lies in
$\{0,\pm1/2\}$. Taking $d=1$, $U=v=\frac12(1,1,1,0,1,0,0,0)^\top$ and
$b=\frac12(0,0,0,1,0,1,1,1)^\top$, we compare in the table below the exact
$\E_\sigma[g_j^2]$ with the sum over identical index tuples, which is the rule
of~(10) and~(17) applied to $g_j$.
\begin{center}
\begin{tabular}{cccc}
  \hline
  $j$ & $\E_\sigma[g_j^2]$ & sum over identical tuples & ratio\\
  \hline
  $1$ & $9/512$ & $15/1024$ & $6/5$\\
  $2$ & $45/8192$ & $45/32768$ & $4$\\
  $3$ & $369/131072$ & $135/1048576$ & $328/15$\\
  \hline
\end{tabular}
\end{center}
The ratio grows with $j$, and we saw the same growth on random instances with
$n=16$. Our computations show no stated bound of~\cite{psw17} failing. The
Gaussian case of~\cite{psw17}, their Theorem~7, does not use Lemma~9, since
there $SU$ and $Sb$ are independent when $U^\top b=0$~\cite[pp.~6--7]{psw17}.
Our fix is a different fast sketch: in Theorem~\ref{thm:main}, we give a
complete proof of the coordinate-wise guarantee for the balanced Gaussian-pooled
transform of Section~\ref{sec:balanced_gaussian_pooled_transform}. That result
says nothing about the SRHT itself, and the statement of Theorem~10
of~\cite{psw17} may still be true.

\subsection{The Regression Bound of Song, Ye, Yin and Zhang}
\label{sec:syyz23_gap}

This subsection describes the dependency bug in the argument of~\cite{syyz23},
shows that the resulting gap is genuine, and shows that a Gaussian diagonal in
the usual position does not repair it.

\subsubsection{The Original Dependency Bug}
\label{sec:original_dependency_bug}

Take a thin singular value decomposition $A=U\Sigma V^\top$
and define the least-squares residual
$
  r:=b-Ax^\star.
$
By the normal equations, we have $U^\top r=0$. Whenever $SU$ has full
column rank, the error of the sketched solution satisfies
\[
  \widehat{x}-x^\star
  =
  V\Sigma^{-1}(U^\top S^\top SU)^{-1}U^\top S^\top Sr.
\]
Left-multiplying this identity by $a^\top$ and using the notation defined
below, we obtain
\begin{equation}
  a^\top(\widehat{x}-x^\star)=c^\top G^{-1}h,
  \label{eq:contrast}
\end{equation}
where
\[
  c:=\Sigma^{-1}V^\top a,
  \qquad
  G:=U^\top S^\top SU,
  \qquad
  h:=U^\top S^\top Sr.
\]
Define $u(S):=UG^{-1}c=U(U^\top S^\top SU)^{-1}c$, which satisfies
$c^\top G^{-1}h=u(S)^\top S^\top Sr$ because the sketched Gram matrix $G$ is
symmetric. Since $U^\top r=0$, we also have $u(S)^\top r=0$. With these
identities, Eq.~\eqref{eq:contrast} takes the exact OCE form
$
  c^\top G^{-1}h=u(S)^\top S^\top Sr=u(S)^\top(S^\top S-I)r.
$

\begin{definition}[Oblivious coordinate-wise embedding,~\cite{syyz23}]
\label{def:oce}
Fix $\beta>0$, $0<\delta<1$, and positive integers $m,n$. A distribution
$\mathcal{D}$ over matrices $S\in\R^{m\times n}$ is a
$(\beta,\delta,n)$-oblivious coordinate-wise embedding (OCE) if, for every
fixed pair $g,h\in\R^n$ chosen independently of $S$, a draw
$S\sim\mathcal{D}$ satisfies
\[
  |g^\top(S^\top S-I_n)h|
  \leq \frac{\beta}{\sqrt{m}}\,\|g\|_2\|h\|_2
\]
with probability at least $1-\delta$.
\end{definition}

The original proof of~\cite{syyz23} applies Definition~\ref{def:oce} to the
pair $(u(S),r)$. This application is not allowed because $u(S)$ depends on the
same random sketch $S$. The definition has the quantifiers
\[
  \text{for every fixed $g,h$,}\qquad
  \Pr_S[\text{failure for $g,h$}]\leq\delta,
\]
not
\[
  \Pr_S[\text{failure for an adaptively chosen $u(S),r$}]\leq\delta.
\]
Conditioning first on the event that $S$ embeds the column space of $A$
(Definition~\ref{def:ose}) does not help: it changes the law of $S$, while
$u(S)$ still depends on the remaining randomness. A norm bound on $u(S)$ does
not make it fixed, independent, or covered by a uniform OCE statement.

\subsubsection{The Gap Is Genuine, Not Merely Formal}
\label{sec:gap_genuine}

With the following compact example, we show that OSE plus fixed-vector OCE
does not imply the adaptive inverse-Gram conclusion.

\begin{proposition}[OSE and fixed-vector OCE are insufficient]
For every sufficiently large $d$, there is a distribution $\mathcal{D}_d$ on
$S\in\R^{(d+1)\times(d+1)}$ such that all singular values of $S$
lie in $[0.9,1.1]$ deterministically and, for every fixed
$g,h\in\R^{d+1}$ and every $0<\delta<1$,
\[
  \Pr_{S\sim\mathcal{D}_d}[
    |g^\top(S^\top S-I_{d+1})h|
    >
    C\sqrt{\frac{\log(2/\delta)}{d}}\,\|g\|_2\|h\|_2
  ]
  \leq\delta.
\]
Equivalently, after changing $C$ by an absolute factor,
$\mathcal{D}_d$ is a
$(C\sqrt{\log(2/\delta)},\delta,d+1)$-OCE according to
Definition~\ref{def:oce}. Nevertheless, a fixed one-dimensional regression
contrast is bounded below by an absolute constant with probability one.
\end{proposition}

\begin{proof}
Let $U:=[e_1,\ldots,e_d]$, draw $v$ uniformly from the unit sphere in
$\operatorname{span}(e_2,\ldots,e_d)$, and fix
$\rho:=1/20$ and $\alpha:=1/20$. We define
\[
  M(v):=
  \begin{bmatrix}
    B(v) & \alpha v\\
    \alpha v^\top & 1
  \end{bmatrix},
  \qquad
  B(v):=I_d-\rho(e_1v^\top+ve_1^\top),
  \qquad
  S:=M(v)^{1/2}.
\]
Let $\mathcal{D}_d$ denote the resulting distribution of $S$.
On $\operatorname{span}(e_1,v,e_{d+1})$, the perturbation $M-I$ has
eigenvalues $0$ and $\pm\sqrt{\rho^2+\alpha^2}$, and it vanishes on the
orthogonal complement. Since $\sqrt{\rho^2+\alpha^2}<1/10$, the matrix $M$
is positive definite. Moreover, $S^\top S=M$, so the singular values of $S$
are
\[
  1,\qquad
  \sqrt{1-\sqrt{\rho^2+\alpha^2}},\qquad
  \sqrt{1+\sqrt{\rho^2+\alpha^2}},
\]
together with additional copies of $1$. These values all lie in $[0.9,1.1]$.

Let $\mathcal{H}:=\operatorname{span}(e_2,\ldots,e_d)$ and let
$P_{\mathcal{H}}$ be the orthogonal projector onto $\mathcal{H}$. For fixed
$g,h\in\R^{d+1}$, a direct expansion gives
\[
  g^\top(M-I)h=v^\top w,
\]
where the fixed vector
\[
  w:=P_{\mathcal{H}}[
    -\rho(g_1h_{1:d}+h_1g_{1:d})
    +\alpha(h_{d+1}g_{1:d}+g_{d+1}h_{1:d})
  ]
\]
satisfies
\[
  \|w\|_2\leq 2(\rho+\alpha)\|g\|_2\|h\|_2.
\]
Consequently, by spherical concentration~\cite[Lemma~2.2]{dg03}, there is a
universal constant $C>0$ such that, for every $0<\delta<1$,
\[
  \Pr[
    |g^\top(S^\top S-I_{d+1})h|
    >
    C\sqrt{\frac{\log(2/\delta)}{d}}\,\|g\|_2\|h\|_2
  ]
  \leq\delta.
\]
Since the sketch has $m=d+1$ rows, increasing $C$ by an absolute factor shows
that $\mathcal{D}_d$ is a
$(C\sqrt{\log(2/\delta)},\delta,d+1)$-OCE according to
Definition~\ref{def:oce}.

Finally, we take $A:=U$, $b:=e_{d+1}$, $r:=e_{d+1}$, and $a:=e_1$. Then
$U^\top r=0$, so the least-squares minimizer is $x^\star=0$. Since
$S^\top S=M$, we have
\[
  (SA)^\dagger Sb=(U^\top MU)^{-1}U^\top Mr=B(v)^{-1}(\alpha v).
\]
On the ordered basis $(e_1,v)$ of $\operatorname{span}(e_1,v)$, we have
\[
  B(v)=
  \begin{bmatrix}
    1&-\rho\\
    -\rho&1
  \end{bmatrix},
  \qquad
  B(v)^{-1}
  =\frac{1}{1-\rho^2}
  \begin{bmatrix}
    1&\rho\\
    \rho&1
  \end{bmatrix}.
\]
Consequently,
\[
  |e_1^\top(SA)^\dagger Sb|
  =\frac{\alpha\rho}{1-\rho^2}.
\]
This value is independent of $d$, while the claimed fixed-OCE conclusion
would tend to zero as $d\to\infty$.
\end{proof}

\subsubsection{Why a Gaussian Diagonal in the Usual Position Also Fails}
\label{sec:gaussian_diagonal_fails}

One might replace the Rademacher diagonal in $m^{-1/2}PHD$ by a Gaussian
diagonal, but this change does not repair the argument. For an unnormalized
Hadamard matrix $H$ and a matrix $P$ that samples $m$ rows independently and
uniformly, put
\[
  S:=m^{-1/2}PHD_\gamma,
  \qquad
  K:=m^{-1}H^\top P^\top PH.
\]
Here the Gaussian diagonal $D_\gamma:=\operatorname{diag}(\gamma_1,\ldots)$
has entries $\gamma_j\overset{\mathrm{iid}}{\sim}\mathcal{N}(0,1)$ drawn
independently of the row sampling $P$. For two distinct Hadamard columns, we
have $K_{11}=K_{22}=1$ and
\[
  \rho:=K_{12}=\frac{1}{m}\sum_{k=1}^m\zeta_k
\]
for independent Rademacher variables $\zeta_k$. With
$
  u:=\frac{e_1+e_2}{\sqrt{2}}, 
  r:=\frac{e_1-e_2}{\sqrt{2}},
$
we have $u^\top r=0$, but the one-dimensional sketched regression coefficient is
\begin{equation}
  \frac{u^\top S^\top Sr}{u^\top S^\top Su}
  =
  \frac{\gamma_1^2-\gamma_2^2}
       {\gamma_1^2+\gamma_2^2+2\rho\gamma_1\gamma_2}.
  \label{eq:gaussian_diagonal_failure}
\end{equation}
Hoeffding's inequality gives $\Pr[|\rho|>1/2]\leq 2e^{-m/8}$. Independently of
the row sampling, the ratio $R_\gamma:=\gamma_1/\gamma_2$ is a standard Cauchy
random variable.
Since the two events $|R_\gamma|\geq 2$ and $|R_\gamma|\leq 1/2$ are disjoint, the Cauchy law gives
\[
  \Pr[\max(\gamma_1^2,\gamma_2^2)\geq 4\min(\gamma_1^2,\gamma_2^2)]=\Pr[|R_\gamma|\geq 2]+\Pr[|R_\gamma|\leq 1/2]=2-\frac{4}{\pi}\arctan 2.
\]
To verify the claimed constant, put
$a:=\max(\gamma_1^2,\gamma_2^2)$ and
$b:=\min(\gamma_1^2,\gamma_2^2)$, which satisfy $a\geq 4b$ on the imbalance
event. If also $|\rho|\leq 1/2$, then the numerator of the coefficient
satisfies
\[
  |\gamma_1^2-\gamma_2^2|=a-b\geq\frac{3a}{4},
\]
while its denominator satisfies
\[
  |\gamma_1^2+\gamma_2^2+2\rho\gamma_1\gamma_2|
  \leq a+b+\sqrt{ab}
  \leq\frac{7a}{4}.
\]
Thus the absolute value in Eq.~\eqref{eq:gaussian_diagonal_failure} is at
least $3/7$, and since the two conditions hold independently, this happens with
probability at least
\[
  (2-\frac{4}{\pi}\arctan 2)(1-2e^{-m/8})
\]
for arbitrarily large $m$. In our construction, we instead place the Gaussian
weights \emph{after} norm flattening, relying on disjoint pooling to supply
the independence that drives our analysis.

\section{The Balanced Gaussian-Pooled Transform}
\label{sec:balanced_gaussian_pooled_transform}

We begin with the formal definition of our sketch and then explain the ideas
behind its analysis in the remainder of this section.

\begin{definition}[Balanced Gaussian-pooled sketch]
\label{def:balanced_gaussian_pooled_sketch}
Round the desired output row count $m$ up to a power of two, fix
$\kappa\in(0,1/2)$, and define
$
  \Lambda:=\log(Cnmd/(\kappa\delta)).
$
Choose $s$ to be the smallest power of two satisfying
\[
  s\geq
  \max\{
    \lceil\frac{n}{m}\rceil,
    C\kappa^{-2}(d+1)\Lambda^2
  \},
  \qquad
  N:=ms.
\]
Then $N$ is a power of two and $N\geq n$. Let
$J:\R^n\to\R^N$ append zeros and let
$F_N:=H_N/\sqrt{N}$ be the normalized orthogonal Hadamard matrix.

Draw independently:
\begin{itemize}
  \item a Rademacher diagonal $D_\sigma\in\R^{N\times N}$;
  \item a uniform permutation matrix $\Pi\in\R^{N\times N}$;
  \item a Gaussian diagonal
  $D_\gamma:=\operatorname{diag}(\gamma_1,\ldots,\gamma_N)$, with
  $\gamma_j\overset{\mathrm{iid}}{\sim}\mathcal{N}(0,1)$.
\end{itemize}
Partition $[N]$ into $m$ consecutive blocks ${\cal B}_1,\ldots,{\cal B}_m$, each of
size $s$. Let $B\in\{0,1\}^{m\times N}$ sum the coordinates in each block, i.e.,
$B_{i,j}:=1$ if $j\in{\cal B}_i$ and $0$ otherwise. Finally, define the
balanced Gaussian-pooled sketch $S\in\R^{m\times n}$ by
\[
  \boxed{S:=BD_\gamma\Pi F_ND_\sigma J.}
\]
The regression solver sees only $SA$ and $Sb$; there is no preliminary
regression and no post-compression of another regression solution.
\end{definition}

While the transform is dense in the original coordinates almost surely, we
apply it in factored form by padding, flipping signs, taking one Hadamard
transform, permuting, multiplying by Gaussians, and summing blocks.

As we saw in Sections~\ref{sec:original_dependency_bug}
and~\ref{sec:gaussian_diagonal_fails}, the main obstacle is not subspace
preservation alone but the dependence created by the inverse sketched Gram
matrix. With $A=U\Sigma V^\top$ and with $r:=b-Ax^\star$, $X:=SU$,
$y:=Sr$, $M:=X^\top X$, and $c:=\Sigma^{-1}V^\top a$, we obtain from the
normal equations that $a^\top(\widehat{x}-x^\star)=c^\top M^{-1}X^\top y$
whenever $M$ is invertible. Since the direction $UM^{-1}c$ depends on the
same sketch, a fixed-vector OCE bound cannot be applied with it, which is
precisely the failure mode described in
Section~\ref{sec:original_dependency_bug}. In our proof, we instead expose the
randomness in stages so that $M^{-1}$ may depend on the design while, after
fixing the first-stage randomness, the remaining noise is independent of the
design.

\paragraph{Stage 1: balanced block geometry.}
We first expose the factor $Q:=\Pi F_ND_\sigma$ in the factorization
$S=BD_\gamma QJ$ of the sketch in
Definition~\ref{def:balanced_gaussian_pooled_sketch}. For $r\ne 0$, we apply
$Q$ to the padded augmented orthonormal system $[JU,Jr/\|r\|_2]$ and show in
Lemma~\ref{lem:simultaneous_block_embedding} that every balanced block is
nearly isotropic simultaneously. The core lemma handles the case $r=0$
separately. Claim~\ref{clm:blockwise_covariance_bounds} translates this event
into a nearly uniform block design covariance $C_i$, a small
design--residual covariance $h_i$, and controlled residual variance $v_i$.
Orthogonality and $U^\top r=0$ also give the exact identities
$\sum_{i=1}^m C_i=I_d$ and $\sum_{i=1}^m h_i=0$ in
Claim~\ref{clm:block_sum_identities}.

\paragraph{Stage 2: exact conditional Gaussian regression.}
We fix a good realization of $Q$ and then expose the Gaussian diagonal
$D_\gamma$. Because the pooling blocks are disjoint, we obtain from
Claim~\ref{clm:conditional_gaussian_pooling} independent, centered, jointly
Gaussian row pairs $(x_i,y_i)$ across $i$, together with their exact
conditional covariances. We assemble these rows into the conditional design
matrix $X$, the response vector $y$, and the conditional Gram matrix $M$ in
Definition~\ref{def:conditional_design_gram}. In
Lemma~\ref{lem:exact_conditional_regression}, we residualize each pair as
$y_i=x_i^\top\theta_i+\xi_i$ and prove that, conditional on $Q$, the entire
noise vector $\xi$ is independent of the entire design $X$. By the
cancellation and score-decomposition parts of the same lemma, we have
$X^\top y=Z+X^\top\xi$, where
$Z:=\sum_{i=1}^m(x_ix_i^\top-C_i)\theta_i$. Thus $Z$ is a conditionally
centered, design-dependent fluctuation, while $X^\top\xi$ is driven by a
Gaussian noise vector $\xi$ that is independent of $X$ conditional on $Q$.
This exact separation of the bias term from the noise term is the key repair.

\paragraph{Stage 3: concentration and assembly.}
We show in Lemma~\ref{lem:design_bias_concentration} that, uniformly for every
good $Q$, the matrix $M$ is well conditioned and $Z$ is small. We then
condition further on $(Q,X)$: the quantities $M^{-1}$ and $Z$ are fixed, while
$c^\top M^{-1}X^\top\xi$ is a centered Gaussian with controlled variance.
The choice $\kappa=c_0/\sqrt d$ supplies the $d^{-1/2}$ factor for the bias,
while the row count controls the Gaussian noise.
Lemma~\ref{lem:adaptive_safe_core} then combines the block, design-and-bias,
and noise events by the tower property. Finally,
Theorem~\ref{thm:one_shot_coordinate_regression} transfers the resulting
core estimate to least squares and obtains the simultaneous coordinate bound
by a union bound. Figure~\ref{fig:conditional_probability_flow} records the
order of conditioning in our proof and the associated failure budget.

\begin{figure}[!ht]
  \centering
  \includegraphics[
    width=\linewidth,
    height=0.68\textheight,
    keepaspectratio
  ]{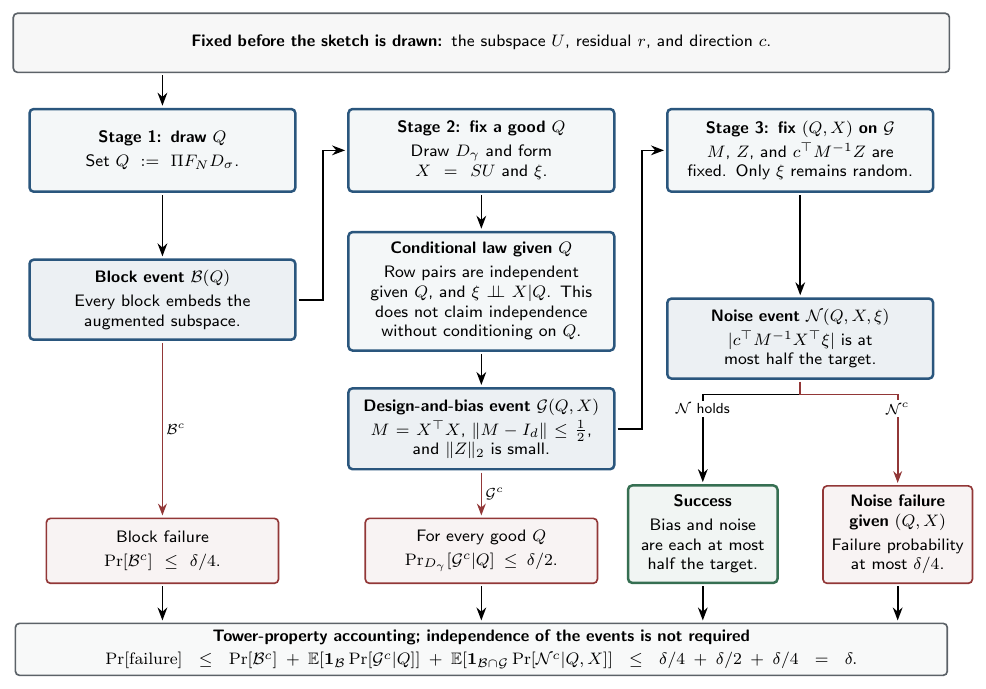}
  \caption{Conditioning hierarchy in the proof. After $Q$ is fixed,
  $D_\gamma$ generates $(X,\xi)$, with $\xi\ind X$ conditional on $Q$.
  The design-and-bias event is measurable with respect to $(Q,X)$; after
  conditioning on $(Q,X)$, only $\xi$ remains random for the noise bound.
  The failure probabilities are combined by the tower property, not by
  independence of the events.}
  \label{fig:conditional_probability_flow}
\end{figure}

\subsection{Conditional Gaussian Pooling}
\label{sec:conditional_pooling}

The first property we extract from the construction is distributional:
conditional on the Hadamard-and-permutation stage, the pooled rows are
independent Gaussian vectors with explicit covariances.

\begin{claim}[Conditional Gaussian pooling and exact block covariances]
\label{clm:conditional_gaussian_pooling}
Put $Q:=\Pi F_ND_\sigma$. Conditional on $Q$, for every $p\geq 1$ and
fixed $T\in\R^{N\times p}$, define
\[
  Y_T:=BD_\gamma QT,\qquad
  \bar T_i:=(QT)_{{\cal B}_i,\ast}\in\R^{s\times p}.
\]
Conditional on $Q$, the rows of $Y_T$ are mutually independent centered
Gaussian vectors, and, for every $i\in[m]$,
\[
  \operatorname{Cov}_\gamma((Y_T)_{i,\ast}^\top|Q)
  =\bar T_i^\top\bar T_i.
\]
Moreover,
\[
  \E_\gamma[Y_T^\top Y_T|Q]=T^\top T.
\]
In particular, taking $T=J$ gives $Y_J=S$; hence, conditional on $Q$, the
rows of $S$ are mutually independent centered Gaussian vectors and
$
  \E_\gamma[S^\top S|Q]=I_n.
$ 
\end{claim}
\begin{proof}
For each $i\in[m]$, put $g_i:=(\gamma_j)_{j\in{\cal B}_i}$. Since $B$
sums the coordinates in each block,
$
  (Y_T)_{i,\ast}^\top=\bar T_i^\top g_i
$.
Conditional on $Q$, the vectors $g_i$ are mutually independent standard
Gaussian vectors because the blocks ${\cal B}_1,\ldots,{\cal B}_m$ are
disjoint. These facts prove the asserted conditional Gaussian law,
independence, and covariance formula. Since the blocks partition $[N]$ and
$Q$ is orthogonal, we have
\[
  \E_\gamma[Y_T^\top Y_T|Q]=\sum_{i=1}^m \bar T_i^\top\bar T_i=T^\top Q^\top QT=T^\top T.
\]
Taking $T=J$ and using $J^\top J=I_n$, we obtain the final assertion.
\end{proof}

The isotropy identity alone is not a concentration statement. The stronger
property we use in the regression decomposition below is the conditional row
independence.

\section{Proof of the Corrected Core Lemma}
\label{sec:core_lemma_proof}

This section carries out the three stages of
Section~\ref{sec:balanced_gaussian_pooled_transform} and assembles them into the
corrected core lemma.

\subsection{Block Geometry Supplied by the Randomized Hadamard Transform}
\label{sec:block_geometry}

The role of the Hadamard stage is to make every block a very accurate embedding
of the one fixed augmented subspace relevant to the regression. As we show next,
the parameter choice in Definition~\ref{def:balanced_gaussian_pooled_sketch}
implies Eq.~\eqref{eq:block_embedding_condition} after increasing the universal
constant $C$.

\begin{lemma}[Simultaneous block embedding]
\label{lem:simultaneous_block_embedding}
Let $0<\delta<1/2$, and let $W\in\R^{N\times p}$ be fixed
independently of $(D_\sigma,\Pi)$ and have orthonormal columns, with
$1\leq p\leq d+1$. Put $Q:=\Pi F_ND_\sigma$. Let
$W_i:=(QW)_{{\cal B}_i,\ast}$. If
\begin{equation}
  s\geq
  C\kappa^{-2}p\log^2(Npm/\delta),
  \label{eq:block_embedding_condition}
\end{equation}
then, with probability at least $1-\delta$ over $(D_\sigma,\Pi)$, the event
\[
  \mathcal{E}_{\mathrm{block}}(W)
  :=
  \{
    \|mW_i^\top W_i-I_p\|\leq\kappa
    \text{ for every }i\in[m]
  \}
\]
occurs.
Moreover, if $n\geq d$ and $s,N,\Lambda$ are chosen as in
Definition~\ref{def:balanced_gaussian_pooled_sketch}, then
Eq.~\eqref{eq:block_embedding_condition} holds after increasing the universal
constant $C$, and
\[
  N=O(n+m\kappa^{-2}(d+1)\Lambda^2).
\]
\end{lemma}

\begin{proof}
Put $Z_0:=F_ND_\sigma W$ and let $z_j:=(Z_0)_{j,\ast}^\top$. Since
$F_ND_\sigma$ is orthogonal and $W$ has orthonormal columns, we have
$
  \sum_{j=1}^N z_jz_j^\top=I_p
  $.
By randomized-Hadamard flattening~\cite[Lemma~3.3]{t11}, except with
probability $\delta/2$, we have
\begin{equation}
  \max_{j\in[N]}\|z_j\|_2^2
  \leq
  \frac{C(p+\log(2N/\delta))}{N}.
  \label{eq:flattening}
\end{equation}
We fix $D_\sigma$ for which Eq.~\eqref{eq:flattening} holds. For a fixed block
$i$, the set $T_i:=\Pi^{-1}{\cal B}_i$ is a uniform size-$s$ subset of
$[N]$, and
\[
  W_i^\top W_i=\sum_{j\in T_i}z_jz_j^\top.
\]
Then we have
$
  \E_\Pi[W_i^\top W_i|D_\sigma]
  =\frac{s}{N}\sum_{j=1}^N z_jz_j^\top
  =\frac{1}{m}I_p
$.
Moreover, Eq.~\eqref{eq:flattening} bounds every positive-semidefinite
summand by
\[
  \|z_jz_j^\top\|
  \leq\frac{C(p+\log(2N/\delta))}{N}.
\]
By the matrix Chernoff bound for sampling without
replacement~\cite[Theorem~2.2]{t11}, which rests on the comparison principle
of Gross and Nesme~\cite{gn10}, we obtain
\[
  \Pr_\Pi[
      \|mW_i^\top W_i-I_p\|>\kappa
    \,|\,D_\sigma
  ]
  \leq
  2p\exp(
    -\frac{c\kappa^2s}{p+\log(2N/\delta)}
  ).
\]
For completeness, let $L_0:=\log(Npm/\delta)$. In the present parameter
range $L_0\geq\log 2$, and
\[
  \frac{pL_0^2}{p+\log(2N/\delta)}\geq cL_0.
\]
Thus Eq.~\eqref{eq:block_embedding_condition}, with a sufficiently large
universal constant, makes the last probability at most $\delta/(2m)$.
A union bound over the $m$ blocks and the flattening failure proves that
$\mathcal{E}_{\mathrm{block}}(W)$ occurs with the stated probability, even
though the blocks need not be independent.

For the final assertion, put
$K:=C\kappa^{-2}(d+1)\Lambda^2$, which satisfies $K\geq1$ after increasing
the universal constant $C$. By the minimality of the power-of-two choice of
$s$ in Definition~\ref{def:balanced_gaussian_pooled_sketch}, we have
\[
  N=ms
  <2m\max\{\lceil n/m\rceil,K\}
  \leq2n+2m+2mK
  =O(n+mK),
\]
where the last step follows from $m\lceil n/m\rceil\leq n+m$ and $K\geq1$.
Since $n\geq d$, $p\leq d+1$, and
$\Lambda=\log(Cnmd/(\kappa\delta))$, we obtain from this bound that
$\log(Npm/\delta)=O(\Lambda)$. Therefore, after increasing $C$ once more,
$s\geq C\kappa^{-2}(d+1)\Lambda^2$ implies
Eq.~\eqref{eq:block_embedding_condition}.
\end{proof}

We assume for now that $r\ne 0$ and identify $U,r$ with their padded images
$JU,Jr$. Since $U^\top U=I_d$ and $U^\top r=0$, the matrix
$
  W:=[U,r/\|r\|_2]\in\R^{N\times(d+1)}
$
has orthonormal columns. We apply Lemma~\ref{lem:simultaneous_block_embedding}
to this $W$ and abbreviate
$\mathcal{E}_{\mathrm{block}}:=\mathcal{E}_{\mathrm{block}}(W)$.
We put $
  \bar U:=QU,$ $\bar r:=Qr,
$
and, for every $i\in[m]$, define the block restrictions
$
  \bar U_i:=(\bar U)_{{\cal B}_i,\ast}\in\R^{s\times d},
  $$
  \bar r_i:=(\bar r)_{{\cal B}_i}\in\R^{s}
$.
\begin{claim}[Blockwise covariance bounds]
\label{clm:blockwise_covariance_bounds}
Under the standing assumption $r\ne 0$, for every $i\in[m]$, define
\[
  C_i:=\bar U_i^\top\bar U_i,\qquad
  h_i:=\bar U_i^\top\bar r_i,\qquad
  v_i:=\|\bar r_i\|_2^2.
\]
On $\mathcal{E}_{\mathrm{block}}$, the following hold simultaneously for every $i\in[m]$:
\begin{list}{(\alph{enumi})}{
  \usecounter{enumi}
  \renewcommand{\theenumi}{\alph{enumi}}
  \setlength{\leftmargin}{2.5em}
  \setlength{\labelwidth}{2em}
  \setlength{\labelsep}{0.5em}
  \setlength{\itemsep}{0.2em}
  \setlength{\parsep}{0pt}
  \setlength{\topsep}{0.3em}
}
  \item \label{part:block_design_covariance_bound}
  $\frac{1-\kappa}{m}I_d\preceq C_i\preceq\frac{1+\kappa}{m}I_d$.

  \item \label{part:block_cross_covariance_bound}
  $\|h_i\|_2\leq\frac{\kappa}{m}\|r\|_2$.

  \item \label{part:block_response_variance_bound}
  $v_i\leq\frac{1+\kappa}{m}\|r\|_2^2$.
\end{list}
\end{claim}
\begin{proof}
On $\mathcal{E}_{\mathrm{block}}$, we have
$\|mW_i^\top W_i-I_{d+1}\|\leq\kappa$, and the block Gram matrix is
\[
  mW_i^\top W_i-I_{d+1}
  =
  \begin{bmatrix}
    mC_i-I_d & mh_i/\|r\|_2\\
    mh_i^\top/\|r\|_2 & mv_i/\|r\|_2^2-1
  \end{bmatrix}.
\]
The upper-left principal compression satisfies
$\|mC_i-I_d\|\leq\kappa$, which proves
Claim~\ref{clm:blockwise_covariance_bounds}-(\ref{part:block_design_covariance_bound}).
The upper-right rectangular compression satisfies
$\|mh_i/\|r\|_2\|\leq\kappa$, which proves
Claim~\ref{clm:blockwise_covariance_bounds}-(\ref{part:block_cross_covariance_bound}).
The lower-right scalar compression satisfies
$|mv_i/\|r\|_2^2-1|\leq\kappa$, whose upper bound proves
Claim~\ref{clm:blockwise_covariance_bounds}-(\ref{part:block_response_variance_bound}).
\end{proof}
\begin{claim}[Block-sum identities]
\label{clm:block_sum_identities}
The block quantities defined in Claim~\ref{clm:blockwise_covariance_bounds} satisfy,
independently of
$\mathcal{E}_{\mathrm{block}}$,
\[
  \sum_{i=1}^m C_i=I_d,\qquad
  \sum_{i=1}^m h_i=U^\top r=0.
\]
\end{claim}
\begin{proof}
For the covariance sum,
\[
  \sum_{i=1}^m C_i
  =\bar U^\top\bar U
  =U^\top Q^\top QU
  =U^\top U
  =I_d.
\]
The first step follows from the definition of $C_i$ in
Claim~\ref{clm:blockwise_covariance_bounds} and the fact that the blocks form a
partition of the row index set $[N]$, the second step follows from
$\bar U=QU$, the third step follows from $Q^\top Q=I_N$, and the last step
follows from $U^\top U=I_d$.

For the cross-covariance sum,
\[
  \sum_{i=1}^m h_i
  =\bar U^\top\bar r
  =U^\top Q^\top Qr
  =U^\top r
  =0.
\]
The first step follows from the definition of $h_i$ in
Claim~\ref{clm:blockwise_covariance_bounds} and the fact that the blocks form a
partition of the row index set $[N]$, the second step follows from
$\bar U=QU$ and $\bar r=Qr$, the third step follows from $Q^\top Q=I_N$, and
the last step follows from $U^\top r=0$.
\end{proof}

\subsection{Conditional Gaussian Regression: The Adaptive-Safe Step}
\label{sec:conditional_regression}

Throughout this section, we fix a realized $Q$ for which
$\mathcal{E}_{\mathrm{block}}$ occurs and take all probabilities only over
$D_\gamma$.
\begin{definition}[Conditional observations, design, and Gram matrix]
\label{def:conditional_design_gram}
For every $i\in[m]$, put
$g_i:=(\gamma_k)_{k\in {\cal B}_i}\in\R^s$. Conditional on the fixed $Q$,
the vectors $g_1,\ldots,g_m$ are mutually independent and each is distributed
as $\mathcal{N}(0,I_s)$. Define the conditional observation pair by
\[
    x_i:=\bar U_i^\top g_i\in\R^d,\qquad
    y_i:=\bar r_i^\top g_i\in\R.
\]
Equivalently, $x_i^\top$ is the $i$-th row of $SU$ and $y_i$ is the $i$-th
entry of $Sr$. Define the conditional design matrix $X\in\R^{m\times d}$ by
\[
    X_{i,\ast}:=x_i^\top\qquad\text{for every }i\in[m],
\]
and define the response vector and conditional Gram matrix by
\[
    y:=(y_1,\ldots,y_m)^\top\in\R^m,\qquad
    M:=X^\top X=\sum_{i=1}^m x_ix_i^\top\in\R^{d\times d}.
\]
\end{definition}
We apply Claim~\ref{clm:conditional_gaussian_pooling} with
$T:=[U,r]\in\R^{N\times(d+1)}$, using the padded identification above. Since
$(QT)_{{\cal B}_i,\ast}=[\bar U_i,\bar r_i]$, the pairs $(x_i,y_i)$ are,
conditional on this fixed $Q$, mutually independent across
$i$ and centered jointly Gaussian with covariance
\[
    \begin{bmatrix}
        C_i & h_i\\
        h_i^\top & v_i
    \end{bmatrix}.
\]

\begin{lemma}[Exact conditional regression decomposition]
\label{lem:exact_conditional_regression}
For the standing $r\ne 0$ construction and the fixed $Q$ satisfying
$\mathcal{E}_{\mathrm{block}}$,
Claim~\ref{clm:blockwise_covariance_bounds}-(\ref{part:block_design_covariance_bound})
implies that every $C_i$ is positive definite. Define
\[
    \theta_i:=C_i^{-1}h_i,\qquad
    \tau_i^2:=v_i-h_i^\top C_i^{-1}h_i.
\]
Define $\xi_i:=y_i-x_i^\top\theta_i$ and
$\xi:=(\xi_1,\ldots,\xi_m)^\top\in\R^m$. Conditional on this fixed $Q$,
the variables $\xi_1,\ldots,\xi_m$ are mutually independent Gaussians with
$\xi_i\sim\mathcal{N}(0,\tau_i^2)$, and the entire vector $\xi$ is
independent of the entire design matrix $X$. The defining identity is
\begin{equation}
    y_i=x_i^\top\theta_i+\xi_i.
\label{eq:conditional_regression}
\end{equation}
Moreover, we have:
\begin{list}{(\alph{enumi})}{
  \usecounter{enumi}
  \renewcommand{\theenumi}{\alph{enumi}}
  \setlength{\leftmargin}{2.5em}
  \setlength{\labelwidth}{2em}
  \setlength{\labelsep}{0.5em}
  \setlength{\itemsep}{0.2em}
  \setlength{\parsep}{0pt}
  \setlength{\topsep}{0.3em}
}
  \item \label{part:conditional_theta_bound}
  $\|\theta_i\|_2\leq 2\kappa\|r\|_2$ for every $i\in[m]$.
  \item \label{part:conditional_variance_bound}
  $\tau_i^2\leq \frac{2}{m}\|r\|_2^2$ for every $i\in[m]$.
  \item \label{part:conditional_cancellation}
  $\sum_{i=1}^m C_i\theta_i=0$.
  \item \label{part:conditional_score_decomposition}
  \label{part:conditional_bias_definition}
  $X^\top y=Z+X^\top\xi$, where $Z:=\sum_{i=1}^m(x_ix_i^\top-C_i)\theta_i$.
\end{list}
\end{lemma}

\begin{proof}
Throughout the proof, we condition on the fixed $Q$. Since $Q$ depends only on
$(D_\sigma,\Pi)$ and is independent of $D_\gamma$, the vectors $g_i$
remain independent standard Gaussian vectors after this conditioning. Once
$Q$ is fixed, the quantities $\bar U_i,\bar r_i,C_i,h_i,v_i,\theta_i$, and
$\tau_i^2$ are deterministic.
By Claim~\ref{clm:blockwise_covariance_bounds}-(\ref{part:block_design_covariance_bound}), we have
\[
    C_i\succeq \frac{1-\kappa}{m}I_d\succ 0.
\]
The first step follows from Claim~\ref{clm:blockwise_covariance_bounds}-(\ref{part:block_design_covariance_bound}), and
the second step follows from $\kappa<1/2$. Thus $C_i$ is invertible and
$\theta_i=C_i^{-1}h_i$ is well defined. By the definitions in the lemma
statement, we have
\[
    \xi_i=y_i-x_i^\top\theta_i
    =y_i-x_i^\top C_i^{-1}h_i.
\]

Conditional on $Q$, since the pair $(x_i,y_i)$ is centered and jointly
Gaussian, the pair $(x_i,\xi_i)$ is also centered and jointly Gaussian.
Moreover,
\[
    \E_\gamma[x_i\xi_i|Q]
    =\E_\gamma[x_i y_i|Q]-\E_\gamma[x_ix_i^\top|Q]\theta_i
    =h_i-C_i\theta_i=0.
\]
The first step follows from $\xi_i=y_i-x_i^\top\theta_i$, the second
step follows from the covariance blocks
$\E_\gamma[x_i y_i|Q]=h_i$ and
$\E_\gamma[x_ix_i^\top|Q]=C_i$, and the last step follows from
$C_i\theta_i=h_i$. Hence $\xi_i$ is independent of $x_i$, because zero
covariance implies independence for jointly Gaussian variables.

For the conditional variance, we have
\[
    \E_\gamma[\xi_i^2|Q]=v_i-\theta_i^\top h_i=v_i-h_i^\top C_i^{-1}h_i=\tau_i^2.
\]
For the first step, expanding $(y_i-x_i^\top\theta_i)^2$ and using the
three covariance blocks gives
$v_i-2\theta_i^\top h_i+\theta_i^\top C_i\theta_i$, which reduces to
$v_i-\theta_i^\top h_i$ because $C_i\theta_i=h_i$. The second step
follows from $\theta_i=C_i^{-1}h_i$, and the last step follows from the
definition of $\tau_i^2$. In particular, $\tau_i^2\geq 0$ and,
conditional on $Q$, $\xi_i\sim\mathcal{N}(0,\tau_i^2)$.

Conditional on $Q$, for distinct indices $i$ and $j$, $(x_i,\xi_i)$ and
$(x_j,\xi_j)$ are functions of the disjoint Gaussian blocks $g_i$ and $g_j$,
so these pairs are independent and,
in particular, the $\xi_i$ are mutually independent. For completeness, the
stacked vector formed by all $x_i$ and $\xi_i$ is jointly Gaussian, since it
is a linear image of the jointly Gaussian vector formed by $g_1,\ldots,g_m$.
If $i\ne j$, block independence and centering give $\E_\gamma[x_j\xi_i|Q]=0$,
and we proved the same identity for $i=j$ above. Thus the cross-covariance
between $\xi$ and the stacked
design vector $(x_1,\ldots,x_m)$ is zero. Joint Gaussianity now implies that
$\xi$ is independent of the entire design matrix $X$. This proves
Eq.~\eqref{eq:conditional_regression} and the asserted independence.

{\bf Proof of Lemma~\ref{lem:exact_conditional_regression}-(\ref{part:conditional_theta_bound})}.
Because inverting a positive-definite matrix reciprocates its eigenvalues,
we obtain from Claim~\ref{clm:blockwise_covariance_bounds}-(\ref{part:block_design_covariance_bound}) that
\[
    \|C_i^{-1}\|\leq \frac{m}{1-\kappa}.
\]
Consequently,
\[
    \|\theta_i\|_2\leq \|C_i^{-1}\|\|h_i\|_2\leq \frac{\kappa}{1-\kappa}\|r\|_2\leq 2\kappa\|r\|_2.
\]
The first step follows from operator-norm submultiplicativity applied to
$\theta_i=C_i^{-1}h_i$, the second step follows from the preceding inverse
bound and Claim~\ref{clm:blockwise_covariance_bounds}-(\ref{part:block_cross_covariance_bound}), and the last step follows from $\kappa<1/2$.

{\bf Proof of Lemma~\ref{lem:exact_conditional_regression}-(\ref{part:conditional_variance_bound})}.
We have
\[
    0\leq \tau_i^2
    =v_i-h_i^\top C_i^{-1}h_i
    \leq v_i
    \leq \frac{1+\kappa}{m}\|r\|_2^2
    \leq \frac{2}{m}\|r\|_2^2.
\]
The first step follows from the fact that $\tau_i^2$ is the conditional
variance computed above, the second step follows from its definition, the
third step follows from $C_i^{-1}\succeq 0$, the fourth step
follows from
Claim~\ref{clm:blockwise_covariance_bounds}-(\ref{part:block_response_variance_bound}), and the last step follows from $\kappa<1/2$.

{\bf Proof of Lemma~\ref{lem:exact_conditional_regression}-(\ref{part:conditional_cancellation})}.
By the definition of $\theta_i$, we have $C_i\theta_i=h_i$. Therefore,
by Claim~\ref{clm:block_sum_identities}, we have
\[
    \sum_{i=1}^m C_i\theta_i
    =\sum_{i=1}^m h_i
    =0.
\]
The first step follows from $C_i\theta_i=h_i$ applied term by term, and the
second step follows from the identity $\sum_{i=1}^m h_i=0$ in
Claim~\ref{clm:block_sum_identities}.

{\bf Proof of Lemma~\ref{lem:exact_conditional_regression}-(\ref{part:conditional_score_decomposition})}.
Using Eq.~\eqref{eq:conditional_regression}, we have
\begin{align*}
    X^\top y
    &=\sum_{i=1}^m x_i y_i =\sum_{i=1}^m x_ix_i^\top\theta_i
      +\sum_{i=1}^m x_i\xi_i
    =\sum_{i=1}^m (x_ix_i^\top-C_i)\theta_i
      +\sum_{i=1}^m C_i\theta_i+X^\top\xi =Z+X^\top\xi.
\end{align*}
The first step follows from the row definitions of $X$ and $y$, the second
step follows from substituting $y_i=x_i^\top\theta_i+\xi_i$ from
Eq.~\eqref{eq:conditional_regression}, the third step follows from adding and
subtracting $\sum_{i=1}^m C_i\theta_i$ and from
$\sum_{i=1}^m x_i\xi_i=X^\top\xi$, and the last step follows from the
definition of $Z$ and the cancellation in
Lemma~\ref{lem:exact_conditional_regression}-(\ref{part:conditional_cancellation}).

\end{proof}

This decomposition is the key repair: since the noise $\xi$ is independent of
the full design after conditioning on $Q$, the inverse $M^{-1}$ may depend
arbitrarily on the design $X$. For the same reason, no fixed-vector OCE statement is
ever applied to a sketch-dependent vector.

\subsection{Concentration of the Design and Centered Bias}
\label{sec:concentration}

We record the standard scalar tools used below, starting with the
Laurent--Massart inequality~\cite[Lemma~1]{lm00}, which implies that if
$G_1,\ldots,G_m\sim\mathcal{N}(0,1)$ are independent and
$a_1,\ldots,a_m\ge 0$, then
\begin{equation}
    \Pr[
        |\sum_{i=1}^m a_i(G_i^2-1)|>t
    ]
    \le 2\exp(
        -c\min\{
            \frac{t^2}{\sum_{i=1}^m a_i^2},
            \frac{t}{\max_{1\leq i\leq m} a_i}
        \}
    ).
\label{eq:weighted_chi_square_tail}
\end{equation}
We also use the standard Gaussian-product
bound~\cite[Lemmas~2.7.7 and~2.7.10]{v18}, which implies that centered jointly
Gaussian $G,H$ satisfy
\begin{equation}
    \|GH-\E[GH]\|_{\psi_1}
    \le C\sqrt{\E[G^2]\,\E[H^2]},
\label{eq:gaussian_product_orlicz}
\end{equation}
where
$
    \|Z\|_{\psi_1}
    :=\inf\{t>0:\E[e^{|Z|/t}]\le 2\}
$.
With the scalar Bernstein inequality~\cite[Theorem~2.8.1]{v18}, we convert
independent $\psi_1$ bounds into the usual subexponential tail: if the $Y_i$
are independent and centered with $\|Y_i\|_{\psi_1}\le K$, then, for
$t\ge 1$,
$
    \Pr[
        |\sum_{i=1}^m Y_i|>CK(\sqrt{mt}+t)
    ]\le 2e^{-t}
$.
Finally, a $(1/4)$-net of $S^{d-1}$ has size at most $9^d$ and controls a
symmetric matrix norm up to factor two. Similarly, a $(1/2)$-net has size at
most $5^d$ and controls a vector norm up to factor two.

\begin{lemma}[Design and bias concentration]
\label{lem:design_bias_concentration}
Fix any $Q$ for which $\mathcal{E}_{\mathrm{block}}$ (see definition in Lemma~\ref{lem:simultaneous_block_embedding}) occurs and let $L:=d+\log(8/\delta)$. Let $M$ be the conditional Gram matrix defined in
Definition~\ref{def:conditional_design_gram}, and let
$Z:=\sum_{i=1}^m(x_ix_i^\top-C_i)\theta_i$ be the quantity defined in
Lemma~\ref{lem:exact_conditional_regression}-(\ref{part:conditional_bias_definition}). If
$m\ge CL$, then, with conditional probability at least $1-\delta/2$ over
$D_\gamma$, the following hold:
\begin{list}{(\alph{enumi})}{
  \usecounter{enumi}
  \renewcommand{\theenumi}{\alph{enumi}}
  \setlength{\leftmargin}{2.5em}
  \setlength{\labelwidth}{2em}
  \setlength{\labelsep}{0.5em}
  \setlength{\itemsep}{0.2em}
  \setlength{\parsep}{0pt}
  \setlength{\topsep}{0.3em}
}
  \item \label{part:design_concentration}
  $\|M-I_d\|\le\frac{1}{2}$.
  \item \label{part:bias_concentration}
  $\|Z\|_2\le C\kappa\|r\|_2
  (\sqrt{\frac{L}{m}}+\frac{L}{m})$.
\end{list}
These bounds are uniform over every fixed $Q$ for which $\mathcal{E}_{\mathrm{block}}$ occurs.
\end{lemma}

\begin{proof}

{\bf Proof of Lemma~\ref{lem:design_bias_concentration}-(\ref{part:design_concentration}).}
Fix a unit vector $u$. By
Claim~\ref{clm:conditional_gaussian_pooling} and
Definition~\ref{def:conditional_design_gram}, conditional on the fixed $Q$,
the vectors $x_1,\ldots,x_m$ are independent and $x_i\sim\mathcal{N}(0,C_i)$.
We define
$
  a_i:=u^\top C_i u,
 $ $
  u^\top x_i=\sqrt{a_i}G_i,
$ 
where $G_1,\ldots,G_m$ are independent standard Gaussians. The first relation
defines $a_i$, which is strictly positive by
Claim~\ref{clm:blockwise_covariance_bounds}-(\ref{part:block_design_covariance_bound})
together with the standing bound $\kappa<1/2$. The second relation then
follows from standardizing the independent Gaussian variables
$u^\top x_i\sim\mathcal{N}(0,a_i)$.

The coefficient bounds are
\[
  \sum_{i=1}^m a_i
  =u^\top(\sum_{i=1}^m C_i)u
  =u^\top u
  =1,
  \qquad
  \max_{1\leq i\leq m}a_i
  \leq\frac{1+\kappa}{m}
  \leq\frac{2}{m},
  \qquad
  \sum_{i=1}^m a_i^2
  \leq
  (\max_{1\leq i\leq m}a_i)\sum_{i=1}^m a_i
  \leq\frac{2}{m}.
\]
In the first chain, the first step follows from the definition of $a_i$, the
second step follows from Claim~\ref{clm:block_sum_identities}, and the last
step follows from $\|u\|_2=1$. In the second chain, the first step follows
from
Claim~\ref{clm:blockwise_covariance_bounds}-(\ref{part:block_design_covariance_bound}),
and the last step follows from $\kappa<1/2$. In the third chain, the first
step uses $a_i\geq0$ term by term, and the last step follows from the
preceding two bounds.

Since $M=X^\top X=\sum_{i=1}^m x_ix_i^\top$, we have
\[
  u^\top(M-I_d)u
  =\sum_{i=1}^m(u^\top x_i)^2-u^\top u
  =\sum_{i=1}^m a_iG_i^2-\sum_{i=1}^m a_i
  =\sum_{i=1}^m a_i(G_i^2-1).
\]
The first step expands the definition of $M$, the second step substitutes
$u^\top x_i=\sqrt{a_i}G_i$ and uses $u^\top u=1=\sum_{i=1}^m a_i$, and the
last step collects the two sums term by term.

By Eq.~\eqref{eq:weighted_chi_square_tail} with weights $a_i$ and threshold
$1/4$, we have
\[
  \Pr_\gamma[
    |u^\top(M-I_d)u|>1/4
    \,|\,Q
  ]
  \leq
  2\exp(-c\min\{\frac{1}{16\sum_{i=1}^m a_i^2},
                      \frac{1}{4\max_{1\leq i\leq m}a_i}\})
  \leq 2e^{-cm}.
\]
The first step applies Eq.~\eqref{eq:weighted_chi_square_tail} to the preceding
identity. In the last step, we use
$1/(16\sum_{i=1}^m a_i^2)\geq m/32$ and
$1/(4\max_{1\leq i\leq m}a_i)\geq m/8$, both from the coefficient bounds, and
absorb the absolute factor into $c$.

This fixed-$u$ bound is uniform over unit vectors $u$. Let
$\mathcal{N}_{1/4}$ be a $1/4$-net of $S^{d-1}$ with size at most $9^d$.
By a union bound, we have
\[
  \Pr_\gamma[
    \max_{u\in\mathcal{N}_{1/4}}|u^\top(M-I_d)u|>1/4
    \,|\,Q
  ]
  \leq 2\cdot 9^d e^{-cm}
  \leq {\delta}/{4}.
\]
The first step combines the preceding fixed-$u$ tail bound with
$|\mathcal{N}_{1/4}|\leq9^d$, and the last step uses
$m\geq C(d+\log(8/\delta))$ with a sufficiently large universal constant
$C$. On the complementary event, since $M-I_d$ is symmetric, we obtain from
the symmetric-matrix net bound that
\[
  \|M-I_d\|
  \leq 2\max_{u\in\mathcal{N}_{1/4}}|u^\top(M-I_d)u|
  \leq\frac{1}{2}.
\]
The first step is the standard $1/4$-net bound for a symmetric matrix, whose
factor is $(1-2\cdot1/4)^{-1}=2$, and the last step uses the defining
inequality of the complementary event. This proves
Lemma~\ref{lem:design_bias_concentration}-(\ref{part:design_concentration})
with conditional failure probability at most $\delta/4$.

{\bf Proof of Lemma~\ref{lem:design_bias_concentration}-(\ref{part:bias_concentration}).}
Conditional on the fixed $Q$, the vectors $x_1,\ldots,x_m$ are independent
centered Gaussians, while $C_i$ and $\theta_i$ are deterministic. For a unit
vector $u$, we define
\[
  Y_i:=u^\top(x_ix_i^\top-C_i)\theta_i
     =(u^\top x_i)(x_i^\top\theta_i)-u^\top C_i\theta_i.
\]
The displayed step expands the matrix product. Since each $Y_i$ depends only
on $x_i$, the variables $Y_1,\ldots,Y_m$ are conditionally independent.
Moreover, we have
\[
  \E_\gamma[Y_i|Q]
  =u^\top(\E_\gamma[x_ix_i^\top|Q]-C_i)\theta_i
  =0.
\]
The first step substitutes the definition of $Y_i$ and uses linearity of
conditional expectation, and the last step uses
$\E_\gamma[x_ix_i^\top|Q]=C_i$. Thus $Y_i$ is a centered product of the
jointly Gaussian linear forms $u^\top x_i$ and $x_i^\top\theta_i$.

The variance of the second linear form satisfies
\[
  \theta_i^\top C_i\theta_i
  =h_i^\top C_i^{-1}h_i
  \leq\|C_i^{-1}\|\|h_i\|_2^2
  \leq\frac{m}{1-\kappa}\frac{\kappa^2}{m^2}\|r\|_2^2
  \leq\frac{2\kappa^2}{m}\|r\|_2^2.
\]
The first step substitutes $\theta_i=C_i^{-1}h_i$ and uses the symmetry of
$C_i$, the second step applies the operator-norm bound to the quadratic form,
the third step follows from
Claim~\ref{clm:blockwise_covariance_bounds}-(\ref{part:block_design_covariance_bound})
and Claim~\ref{clm:blockwise_covariance_bounds}-(\ref{part:block_cross_covariance_bound}),
and the last step follows from $\kappa<1/2$. Similarly, we have
\[
  u^\top C_i u
  \leq\frac{1+\kappa}{m}
  \leq\frac{2}{m}.
\]
The first step uses
Claim~\ref{clm:blockwise_covariance_bounds}-(\ref{part:block_design_covariance_bound})
and $\|u\|_2=1$, and the last step uses $\kappa<1/2$.

By Eq.~\eqref{eq:gaussian_product_orlicz}, we now have
\[
  \|Y_i\|_{\psi_1}
  \leq C\sqrt{u^\top C_i u}
          \sqrt{\theta_i^\top C_i\theta_i}
  \leq\frac{C\kappa}{m}\|r\|_2.
\]
The first step applies Eq.~\eqref{eq:gaussian_product_orlicz} to the centered
Gaussian product defining $Y_i$, and the last step substitutes the preceding
two variance bounds and absorbs absolute numerical factors into $C$.

By the definitions of $Z$ and $Y_i$, we have
\[
  u^\top Z
  =\sum_{i=1}^m u^\top(x_ix_i^\top-C_i)\theta_i
  =\sum_{i=1}^m Y_i.
\]
The first step substitutes the definition of $Z$, and the last step uses the
definition of $Y_i$. Applying the scalar Bernstein inequality conditionally
on $Q$ to these independent centered variables with
$\|Y_i\|_{\psi_1}\leq C\kappa\|r\|_2/m$, we obtain, for every $t\geq1$,
\[
  \Pr_\gamma[
    |u^\top Z|>
    C\kappa\|r\|_2
    (\sqrt{\frac{t}{m}}+\frac{t}{m})
    \,|\,Q
  ]
  \leq 2e^{-t}.
\]
This step is the scalar Bernstein inequality with
$K:=C\kappa\|r\|_2/m$, and substituting this value into the Bernstein
threshold $K(\sqrt{mt}+t)$ gives the displayed scale after we absorb absolute
constants into $C$.

Take $t:=CL$, where $L=d+\log(8/\delta)$, and let
$\mathcal{N}_{1/2}$ be a $1/2$-net of $S^{d-1}$ with size at most $5^d$.
By a union bound, we have
\[
  \Pr_\gamma[
    \max_{u\in\mathcal{N}_{1/2}}|u^\top Z|>
    C\kappa\|r\|_2(\sqrt{\frac{L}{m}}+\frac{L}{m})
    \,|\,Q
  ]
  \leq 2\cdot 5^d e^{-CL}
  \leq\frac{\delta}{4}.
\]
The first step combines the fixed-$u$ Bernstein bound with
$|\mathcal{N}_{1/2}|\leq5^d$ and absorbs absolute factors from $t=CL$ into
$C$, and the last step uses $L=d+\log(8/\delta)$ and a sufficiently large
universal constant $C$. On the complementary event, we obtain from the vector
net bound that
\[
  \|Z\|_2
  \leq 2\max_{u\in\mathcal{N}_{1/2}}|u^\top Z|
  \leq C\kappa\|r\|_2(\sqrt{\frac{L}{m}}+\frac{L}{m}).
\]
The first step is the standard $1/2$-net bound for a vector, and the last step
uses the defining inequality of the complementary event and absorbs the factor
two into $C$. This proves
Lemma~\ref{lem:design_bias_concentration}-(\ref{part:bias_concentration})
with conditional failure probability at most $\delta/4$.

Adding the two conditional failure probabilities, we obtain
$\delta/4+\delta/4=\delta/2$, which proves the lemma. The estimates are uniform
over every fixed good $Q$ because the conditional Gaussian and independence
structure holds for every such $Q$, while all numerical bounds use only the
deterministic inequalities defining $\mathcal{E}_{\mathrm{block}}$.
\end{proof}

With the block geometry and the conditional Gaussian representation in hand,
we can already record a classical property of the balanced Gaussian-pooled
transform: it is an oblivious subspace embedding. While this fact is not
needed for the core lemma in Section~\ref{sec:core_lemma}, we state it here
for completeness and for comparison with prior work.

\begin{definition}[Oblivious subspace embedding,~\cite{sar06}]
\label{def:ose}
Fix $0<\eta<1$, $0<\delta<1$, and $1\leq d\leq n$. A distribution
$\mathcal{D}$ over matrices $S\in\R^{m\times n}$ is an
$(\eta,\delta,d)$-oblivious subspace embedding (OSE) if, for every fixed
$U\in\R^{n\times d}$ with $U^\top U=I_d$ chosen independently of $S$, a
draw $S\sim\mathcal{D}$ satisfies
\[
  (1-\eta)\|z\|_2^2
  \leq \|SUz\|_2^2
  \leq (1+\eta)\|z\|_2^2
  \qquad\text{for every }z\in\R^d
\]
with probability at least $1-\delta$. Equivalently, with the same
probability,
$
  \|U^\top S^\top SU-I_d\|\leq\eta.
$ 
The word \emph{oblivious} means that the distribution $\mathcal{D}$ does not
depend on $U$.  
\end{definition}

\begin{proposition}[The transform is an OSE]
\label{prop:transform_is_ose}
Let $S$ be the balanced Gaussian-pooled transform of
Definition~\ref{def:balanced_gaussian_pooled_sketch} with
$m\geq C\epsilon^{-2}(d+\log(1/\delta))$ for a sufficiently large universal
constant $C$. Then $S$ is an $(\epsilon,\delta,d)$-OSE.
\end{proposition}
\begin{proof}
Fix $U\in\R^{n\times d}$ with $U^\top U=I_d$, independently of $S$, and put
$W:=JU$. We apply Lemma~\ref{lem:simultaneous_block_embedding} to $W$ with
$p=d$ and failure probability $\delta/2$, absorbing the replacement of
$\delta$ by $\delta/2$ into the universal constant. We condition on a
realized $Q$ for which $\mathcal{E}_{\mathrm{block}}(W)$ occurs and define
$C_i:=W_i^\top W_i$. By Claim~\ref{clm:conditional_gaussian_pooling}, the rows
$x_i^\top$ of $SU$ are independent and $x_i\sim\mathcal{N}(0,C_i)$. Since the
blocks partition $[N]$, $Q$ is orthogonal, and $\kappa<1/2$, we have
$
  \sum_{i=1}^m C_i=I_d
$ and $
  0\preceq C_i\preceq\frac{1+\kappa}{m}I_d\preceq\frac{2}{m}I_d
$.
Put $M:=U^\top S^\top SU=\sum_{i=1}^m x_ix_i^\top$. For a fixed unit vector
$u$, the coefficients $a_i:=u^\top C_i u$ satisfy $\sum_{i=1}^m a_i=1$,
$\max_{1\leq i\leq m}a_i\leq2/m$, and
$\sum_{i=1}^m a_i^2\leq2/m$. Hence, for independent standard Gaussians
$G_1,\ldots,G_m$, we have
$
  u^\top(M-I_d)u=\sum_{i=1}^m a_i(G_i^2-1)
$.
Applying Eq.~\eqref{eq:weighted_chi_square_tail} with threshold
$\epsilon/2$ and using $0<\epsilon\leq1$, we obtain
$
  \Pr_\gamma[|u^\top(M-I_d)u|>\epsilon/2|Q]
  \leq2e^{-cm\epsilon^2}
$.
Let $\mathcal{N}_{1/4}$ be a $1/4$-net of the unit sphere with size at most
$9^d$. For a sufficiently large universal constant $C$, by the assumed lower
bound on $m$ and a union bound, we have with conditional probability at
least $1-\delta/2$ that
$
  \|M-I_d\|
  \leq2\max_{u\in\mathcal{N}_{1/4}}|u^\top(M-I_d)u|
  \leq\epsilon
$.
Adding the failure probability of the block event proves the claim. Since the
law of $S$ is independent of $U$, the embedding is oblivious.
\end{proof}

\subsection{The Corrected Core Lemma}
\label{sec:core_lemma}

\begin{lemma}[Adaptive-safe one-shot core lemma]
\label{lem:adaptive_safe_core}
Let $0<\epsilon\le 1$, $0<\delta<1/2$, $U\in\R^{n\times d}$
have orthonormal columns, and let fixed $r\in\R^n$ and
$c\in\R^d$ satisfy $U^\top r=0$. Use the sketch from Definition~\ref{def:balanced_gaussian_pooled_sketch} with
\begin{equation}
    \kappa:=\frac{c_0}{\sqrt d},\qquad
    m\ge C\epsilon^{-2}d\log(8/\delta),
\label{eq:core_parameters}
\end{equation}
where $c_0>0$ is a sufficiently small universal constant; use the
corresponding choices of $s,N$ in Definition~\ref{def:balanced_gaussian_pooled_sketch}. With probability at least $1-\delta$, the matrix $SU$
has full column rank and
\begin{equation}
    |
        c^\top(U^\top S^\top SU)^{-1}U^\top S^\top Sr
    |
    \le \frac{\epsilon}{\sqrt d}\|c\|_2\|r\|_2.
\label{eq:core_bound}
\end{equation}
\end{lemma}

\begin{proof}
We first assume $r\ne 0$ and let $L:=d+\log(8/\delta)$. Since $d\geq 1$
and $\delta<1/2$, we have
$
  L\leq C d\log(8/\delta).
$
Thus the row assumption in Eq.~\eqref{eq:core_parameters} implies $m\geq CL$,
as required by Lemma~\ref{lem:design_bias_concentration}.

We put $W_r:=[JU,Jr/\|r\|_2]$ and apply
Lemma~\ref{lem:simultaneous_block_embedding} to $W_r$, with its failure
parameter set to $\delta/4$. Replacing $\delta$ by $\delta/4$ only adds
$\log 4$ inside the logarithm in
Eq.~\eqref{eq:block_embedding_condition}, which is absorbed by increasing
the universal constant $C$. Fix any realized $Q$ for which
$\mathcal{E}_{\mathrm{block}}(W_r)$ occurs. Conditional on this $Q$,
Lemma~\ref{lem:design_bias_concentration} fails with probability at most
$\delta/2$.

On its success event,
Lemma~\ref{lem:design_bias_concentration}-(\ref{part:design_concentration})
gives $\|M-I_d\|\leq 1/2$, so $M=X^\top X$ is positive definite and
$\|M^{-1}\|\leq 2$. By Definition~\ref{def:conditional_design_gram}, we have $SU=X$ and $Sr=y$.
Consequently,
\[
  U^\top S^\top SU=M,
  \qquad
  U^\top S^\top Sr=X^\top y.
\]
Lemma~\ref{lem:exact_conditional_regression}-(\ref{part:conditional_score_decomposition}) therefore gives
\[
    c^\top M^{-1}X^\top y
    =c^\top M^{-1}Z+c^\top M^{-1}X^\top\xi.
\]
For the first term,
Lemma~\ref{lem:design_bias_concentration}-(\ref{part:bias_concentration}) gives
\[
  |c^\top M^{-1}Z|\leq\|c\|_2\|M^{-1}\|\|Z\|_2\leq 2\|c\|_2\|Z\|_2\leq C\kappa\|c\|_2\|r\|_2(\sqrt{\frac{L}{m}}+\frac{L}{m}).
\]
The row assumption and the displayed bound on $L$ imply
\[
  \frac{L}{m}\leq C\epsilon^2,
  \qquad
  \sqrt{\frac{L}{m}}+\frac{L}{m}\leq C\epsilon,
\]
where we used $0<\epsilon\leq 1$. Substituting
$\kappa=c_0/\sqrt d$ from Eq.~\eqref{eq:core_parameters}, we obtain
\[
  |c^\top M^{-1}Z|
  \leq Cc_0\frac{\epsilon}{\sqrt d}
  \|c\|_2\|r\|_2.
\]
We choose $c_0$ small enough that $Cc_0\leq 1/2$, so the bias term is at most
$\epsilon\|c\|_2\|r\|_2/(2\sqrt d)$.

For the second term, we condition on the pair $(Q,X)$. This conditioning is
essential because, conditional on $Q$, $\xi$ is independent of $X$, whereas
mixing over different values of $Q$ need not preserve independence. The
design-and-bias success
event above is measurable with respect to $(Q,X)$. Given $(Q,X)$, the second
term is a centered Gaussian with variance
\[
  \operatorname{Var}_\gamma[c^\top M^{-1}X^\top\xi|Q,X]=\sum_{i=1}^m\tau_i^2(x_i^\top M^{-1}c)^2\leq\max_{1\leq i\leq m}\tau_i^2\|XM^{-1}c\|_2^2\leq\frac{2\|r\|_2^2}{m}\|XM^{-1}c\|_2^2.
\]
The last step follows from
Lemma~\ref{lem:exact_conditional_regression}-(\ref{part:conditional_variance_bound}).
Since $M=X^\top X$,
\[
  \|XM^{-1}c\|_2^2=c^\top M^{-1}X^\top XM^{-1}c=c^\top M^{-1}c\leq 2\|c\|_2^2.
\]
Thus the conditional variance is at most
$4\|c\|_2^2\|r\|_2^2/m$. The Gaussian tail bound at the target threshold
gives
\[
  \Pr_\gamma[
      |c^\top M^{-1}X^\top\xi|
      >
      \frac{\epsilon}{2\sqrt d}\|c\|_2\|r\|_2
    \,|\,Q,X
  ]
  \leq
  2\exp(-\frac{\epsilon^2m}{32d})
  \leq \frac{\delta}{4}
\]
after enlarging $C$. This estimate is uniform over the measurable success
event. Integrating the conditional estimates and adding the block-embedding,
design-and-bias, and noise failures gives
$
  \frac{\delta}{4}+\frac{\delta}{2}+\frac{\delta}{4}=\delta.
$
This proves Eq.~\eqref{eq:core_bound}, and $M\succ 0$ gives the asserted full
column rank of $SU$.

If $r=0$, we apply Lemma~\ref{lem:simultaneous_block_embedding} to $JU$ with
$p=d$ and failure probability $\delta/4$. Conditional on a good $Q$, we repeat
only the proof of
Lemma~\ref{lem:design_bias_concentration}-(\ref{part:design_concentration}),
whose conditional failure probability
is at most $\delta/4$. Hence $M\succ 0$ with probability at least
$1-\delta/2\geq 1-\delta$. Since $Sr=0$, the numerator in
Eq.~\eqref{eq:core_bound} is identically zero. This proves both assertions in
the remaining case.
\end{proof}

\section{Regression Consequence and Running Time}
\label{sec:consequences}

We now transfer the corrected core lemma to least squares and bound the running
time of our sketch.

\subsection{Regression Consequence and Desired Row Count}
\label{sec:regression_consequence}

\begin{theorem}[One-shot coordinate-wise regression]
\label{thm:one_shot_coordinate_regression}
Let $0<\epsilon\leq 1$, $0<\delta<1/2$, let
$A\in\R^{n\times d}$ have full column rank, let
$b\in\R^n$, and fix $a\in\R^d$. Set
$\kappa:=c_0/\sqrt d$, choose $m$ to be the smallest power of two
satisfying
$
  m\geq C\epsilon^{-2}d\log(8/\delta),
$
and draw $S$ according to Definition~\ref{def:balanced_gaussian_pooled_sketch}. Then, with probability at least $1-\delta$, $SA$ has full column
rank, the sketched minimizer is unique, and
\begin{equation}
  |a^\top(\widehat{x}-x^\star)|
  \leq
  \frac{\epsilon}{\sqrt d}\,
  \|a\|_2\,\|Ax^\star-b\|_2\,\|A^\dagger\|_{\op}.
\label{eq:regression_bound}
\end{equation}
For all coordinates simultaneously, it suffices to take
\begin{equation}
  \boxed{
    m=O(\epsilon^{-2}d\log(d/\delta)).
  }
\label{eq:simultaneous_row_count}
\end{equation}
For this simultaneous statement, replace $\delta$ by $\delta/d$ in the
entire parameter selection, including $m,\Lambda,s,N$.
\end{theorem}

\begin{proof}
Let $A=U\Sigma V^\top$ and $r:=b-Ax^\star$, so that $U^\top r=0$ and
$b=Ax^\star+r$. On the good event of
Lemma~\ref{lem:adaptive_safe_core}, the matrix $SU$ has full column rank. Since
$SA=(SU)\Sigma V^\top$, the matrix $SA$ also has full column rank, so the
sketched minimizer is unique.

The normal equations for the sketched problem give
$
  (A^\top S^\top SA)(\widehat{x}-x^\star)
  =A^\top S^\top Sr.
$
Substituting the thin singular value decomposition and inverting the
nonsingular factors yields
$
  \widehat{x}-x^\star
  =V\Sigma^{-1}(U^\top S^\top SU)^{-1}U^\top S^\top Sr.
$
Applying Lemma~\ref{lem:adaptive_safe_core} with
$c:=\Sigma^{-1}V^\top a$ and using
$
  \|c\|_2
  \leq \|\Sigma^{-1}\|\|a\|_2
  =\|A^\dagger\|\|a\|_2,
$
we obtain Eq.~\eqref{eq:regression_bound}. For simultaneous coordinates,
we replace $\delta$ by $\delta/d$ in all parameters, including
$m,\Lambda,s,N$. We then apply the one-coordinate result to each fixed direction
$a=e_j$ with failure probability $\delta/d$ and take a union bound over $j\in[d]$.
\end{proof}

\subsection{Running Time}
\label{sec:runtime}

\begin{lemma}[Sketch-and-solve running time]
\label{lem:sketch_and_solve_running_time}
Under the assumptions and parameter choices of
Theorem~\ref{thm:one_shot_coordinate_regression}, on its success event the
one-shot estimator $\widehat{x}$ can be computed in the exact-arithmetic model
using
$
  O(Nd\log N+\Tmat(d,m,d)+d^\omega)
$
arithmetic operations, where $\Tmat(a,b,c)$ denotes the cost of multiplying an
$a\times b$ matrix by a $b\times c$ matrix. Here $N=\widetilde{O}(n+\epsilon^{-2} d^3)$. In particular, the general bound
is
$
  \widetilde{O}(nd+\epsilon^{-2}d^4).
$

\end{lemma}

\begin{proof}
With $\kappa=c_0/\sqrt d$ and $p=d+1$, the second term in the block-size requirement of
Definition~\ref{def:balanced_gaussian_pooled_sketch} is $\widetilde{O}(d^2)$.
Combining this bound with Eq.~\eqref{eq:simultaneous_row_count} and the choice of
$s,N$ in Definition~\ref{def:balanced_gaussian_pooled_sketch} gives
$
  N=O(n+md^2\Lambda^2)
  =\widetilde{O}(n+\epsilon^{-2}d^3).
$
Applying $S$ to one vector costs $O(N\log N)$. Applying it to all $d$
columns of $A$ and to $b$ therefore costs $O(N(d+1)\log N)=O(Nd\log N)$,
where the first step sums the cost over the $d+1$ inputs and the last step follows from
$d\geq1$.
The resulting dense least-squares problem has size $m\times d$. In the
exact-arithmetic model, let $Y:=SA$ and $z:=Sb$. On the success event of
Theorem~\ref{thm:one_shot_coordinate_regression}, the matrix $Y$ has full column rank,
so $\widehat{x}$ is the unique solution of
$
  (Y^\top Y)\widehat{x}=Y^\top z.
$
This is the normal equation for the full-column-rank least-squares problem.
The two products can be formed together by multiplying $Y^\top$ by the
$m\times(d+1)$ matrix whose columns are those of $Y$ followed by $z$.
After constant-factor padding in the last dimension, this multiplication costs
$O(\Tmat(d,m,d))$. Solving the resulting nonsingular $d\times d$ system costs
$O(d^\omega)$~\cite{bh74}. This is an exact-arithmetic normal-equations
calculation, not a claim that a backward-stable QR factorization has the
same cost. Adding the three costs established above, those of the sketching,
the product formation, and the linear system, gives
$
  T=O(Nd\log N+\Tmat(d,m,d)+d^\omega).
$
Definition~\ref{def:balanced_gaussian_pooled_sketch} and the parameter choices
of Theorem~\ref{thm:one_shot_coordinate_regression} give
$N=\widetilde{O}(n+\epsilon^{-2}d^3)$ and
$m=\widetilde{O}(\epsilon^{-2}d)$. Consequently,
$\Tmat(d,m,d)=O(md^2)=\widetilde{O}(\epsilon^{-2}d^3)$ by classical matrix
multiplication, while $d^\omega=O(d^3)$. Substituting these bounds for $N$,
$\Tmat(d,m,d)$, and $d^\omega$ into the preceding running-time expression
gives
$
  T=\widetilde{O}(nd+\epsilon^{-2}d^4).
$

\end{proof}

\bibliographystyle{alpha}
\bibliography{ref}

\end{document}